\documentclass[runningheads,envcountsame]{llncs}
\usepackage[utf8]{inputenc}
\RequirePackage[T1]{fontenc}
\usepackage{amssymb, amsmath}
\usepackage{etoolbox}
\AtEndEnvironment{proof}{\phantom{}\qed}
\usepackage{ifmtarg}
\usepackage{tikz}
\usepackage{xspace, soul, enumerate}
\usepackage{lineno}

\usepackage{mathtools}
\usepackage{hyperref}

\usepackage{marvosym}
\usepackage[plain]{fancyref}

\usepackage{enumitem}

\newif\ifcomments
\newif\ifchanges
\commentstrue\changestrue %

\newcommand{\ahhNA}{{\normalfont{\textsc{AHH}}}\xspace}
\newcommand{\ahh}[1]{{\normalfont{\textsc{AHH}$(#1)$}}\xspace}
\newcommand{\querygraph}{\calQ}
\newcommand{\datagraph}{\calD}

\newcommand{\RelName}[1]{\mtext{#1}}
\newcommand{\tpl}{\bar}

\newcommand{\mtext}[1]{\textsc{#1}}
\newcommand{\ins}{\mtext{ins}\xspace}
\newcommand{\del}{\mtext{del}\xspace}

\newcommand{\schema}{\ensuremath{\tau}\xspace}

\newcommand{\domain}{\ensuremath{D}}
\newcommand{\query}{\ensuremath{Q}}

\newcommand{\db}{\ensuremath{\calD}\xspace}%
\newcommand{\inp}{\ensuremath{\calI}\xspace}
\newcommand{\aux}{\ensuremath{\calA}\xspace}%

\newcommand{\ans}{\RelName{Ans}}

\newcommand{\N}{\ensuremath{\mathbb{N}}}

\newcommand{\df}{\ensuremath{\mathrel{\smash{\stackrel{\scriptscriptstyle{
    \text{def}}}{=}}}} \;}

\makeatletter 
\newcommand{\ut}[4]{
  \@ifmtarg{#4}{t^{#1}_{#2}(#3) }{t^{#1}_{#2}(#3; #4)}
}

\newcommand{\ite}[3]{
  \@ifmtarg{#1}{
    \mtext{ITE}
   }{
    \mtext{ITE}\text{$(#1,#2,#3)$}  
  }
}
\makeatother

\newcommand  {\myclass} [1]  {\ensuremath{\textsf{\upshape #1}}}

\newcommand{\StaClass}[1]{\myclass{#1}\xspace}

\newcommand{\DynClass}[1]{\myclass{Dyn#1}\xspace}

\newcommand  {\myproblem} [1] {\normalfont{\textsc{#1}}\xspace}

\newcommand     {\LOGSPACE}     {\StaClass{LOGSPACE}}
\newcommand     {\NL}   {\StaClass{NL}}
\newcommand     {\PTIME}    {\StaClass{PTIME}}
\newcommand     {\NP}   {\StaClass{NP}}
\newcommand     {\LOGCFL}   {\StaClass{LOGCFL}}

\newcommand{\ACz}{\mbox{\myclass{AC}$^0$}\xspace}

\newcommand{\ACo}{\mbox{\myclass{AC}$^1$}\xspace}
\newcommand{\SACo}{\mbox{\myclass{SAC}$^1$}\xspace}

\newcommand{\FO}{\StaClass{FO}}

\newcommand{\MSO}{\StaClass{MSO}}

\newcommand{\CQ}[1][]{\StaClass{CQ}}
\newcommand{\UCQ}[1][]{\StaClass{UCQ}}
\newcommand{\CQneg}[1][]{\StaClass{CQ\ensuremath{^{\mneg}}}}
\newcommand{\UCQneg}[1][]{\StaClass{UCQ\ensuremath{^{\mneg}}}}

\newcommand{\mneg}{\neg} %

\newcommand{\DynFO}{\DynClass{FO}}

\let\llncsproof\proof
\renewcommand{\proof}[1][]{%
  \ifx!#1!\else\renewcommand{\proofname}{#1}\fi
  \llncsproof
}

\newenvironment{proofsketch}{\begin{proof}[Proof sketch]}{\end{proof}}

\providecommand {\calA}      {{\mathcal A}\xspace}

\providecommand {\calC}      {{\mathcal C}\xspace}
\providecommand {\calD}      {{\mathcal D}\xspace}
\providecommand {\calE}      {{\mathcal E}\xspace}

\providecommand {\calH}      {{\mathcal H}\xspace}

\providecommand {\calI}      {{\mathcal I}\xspace}

\providecommand {\calP}      {{\mathcal P}\xspace}

\providecommand {\calQ}      {{\mathcal Q}\xspace}

\providecommand {\calV}      {{\mathcal V}\xspace}

  \newcommand{\changeRule}[3]{\textbf{on change}\ #1\ \textbf{update}\ #2\ \textbf{as}\ #3}

\newcommand{\prog}{\ensuremath{\calP}\xspace}

\ifcomments
\newcommand{\commentbox}[1]{\noindent\framebox{\parbox{0.98\linewidth}{#1}}}

\newcommand{\acomment}[2]{\ \\ \fbox{\parbox{0.98\linewidth}{{\sc #1}: #2}}}
\newcommand{\mcomment}[2]{{\color{blue}(#1)}\footnote{#1: #2}} %
\else
\newcommand{\commentbox}[1]{}
\newcommand{\mcomment}[2]{}
\newcommand{\acomment}[2]{}
\fi

\ifchanges

\setul{}{0.2mm}
\setstcolor{red}

\else

\fi

\newcommand{\weight}{\mtext{w}}
\newcommand{\WG}{\mtext{wg}}
\newcommand{\msfWG}{\mtext{msf}}
\newcommand{\maxarity}{\ensuremath{{a_\mtext{max}}}\xspace}

\begin{document}
\title{The Dynamic Complexity of Acyclic Hypergraph Homomorphisms}

\author{Nils Vortmeier\inst{1} \and Ioannis 
Kokkinis\inst{2}}
\authorrunning{N. Vortmeier, I. Kokkinis}
\institute{University of Zurich, Switzerland\\\email{nils.vortmeier@uzh.ch}
\and
National Technical University of Athens, Greece, and\\
University of the Aegean, Greece\\
\email{ikokkinis@aegean.gr}}

\maketitle              %
\begin{abstract}
Finding a homomorphism from some hypergraph $\querygraph$ (or some relational structure) to another hypergraph $\datagraph$ is a fundamental problem in computer science. We show that an answer to this problem can be maintained under single-edge changes of $\querygraph$, as long as it stays acyclic, in the \DynFO framework of Patnaik and Immerman that uses updates expressed in first-order logic.
If additionally also changes of $\datagraph$ are allowed, we show that it is unlikely that existence of homomorphisms can be maintained in $\DynFO$. 
\end{abstract}
     
    \section{Introduction}\label{section:introduction}
Many important computational problems can be phrased as the question ``is there a homomorphism from $\querygraph$ to $\datagraph
$?'', where $\querygraph$ and $\datagraph$ are hypergraphs, or more generally, relational structures.
Examples include evaluation and minimisation of conjunctive queries \cite{ChandraM77} and solving constraint satisfaction problems, see \cite{FederV98}. 

The problem \myproblem{Hom} -- is there a homomorphism from $\querygraph$ to $\datagraph
$? -- is \NP-complete in its general form. In the static setting it is well understood which restrictions on $\querygraph$ or $\datagraph
$ render the problem tractable \cite{DalmauKV02,Grohe07,HellN90}. 
A particular restriction of great importance in databases is to demand that $\querygraph$ is \emph{acyclic} \cite{beeri1983desirability}.
This restriction of \myproblem{Hom}, we call it the Acyclic Hypergraph Homomorphism problem \ahhNA, can be solved in polynomial time by Yannakakis' algorithm \cite{yannakakis1981algorithms} and is complete for the complexity class \LOGCFL \cite{gottlob2001complexity}, the class of problems that can be reduced in logarithmic space to a context-free language.

We are interested in a \emph{dynamic} setting where the input of a problem is subject to changes. The complexity-theoretic framework
\DynFO for such a dynamic setting was introduced by Patnaik and Immerman \cite{patnaikI97} and it is closely related to a setting of Dong, Su and Topor \cite{DongST95}.  
In this setting, a relational input structure is subject to a sequence of changes, which are usually insertions of single tuples into a relation, or deletions of single tuples from a relation. After each change, additionally stored auxiliary relations are updated as specified by first-order \emph{update formulas}.
The class \DynFO contains all problems for which the update formulas can maintain the answer for the changing input.

With few exceptions, for example in parts of \cite{muoz_et_al}, research in the \DynFO framework takes a \emph{data complexity} viewpoint: all context-free languages \cite{GeladeMS12} and all problems definable in monadic second-order logic \MSO \cite{DattaMSVZ19} are in \DynFO if the context-free language or the \MSO-definable problem is fixed and not part of the input. Every fixed conjunctive query is trivially in \DynFO, as such a query can be expressed in first-order logic and updates defined by first-order formulas can just compute the result from scratch after every change; however, there are also non-trivial maintenance results for fixed conjunctive queries for subclasses of \DynFO \cite{GeladeMS12,Zeume17}. 
The complexity results for \myproblem{Hom} and \ahhNA of~\cite{gottlob2001complexity,yannakakis1981algorithms} are however from a \emph{combined complexity} perspective: both $\querygraph$ and $\datagraph$ are part of the input. 

\paragraph*{Contributions.} In this paper we study the combined complexity of \ahhNA in the dynamic setting. As inputs we allow hypergraphs and general relational structures over some fixed schema $\schema$.

As our main positive result, we show that \ahh{\schema} is in \DynFO for every schema $\schema$, if $\querygraph$ is subject to insertions and deletions of hyperedges but stays acyclic, and $\datagraph$ may initially be arbitrary but is not changed afterwards.
A main building block for this result is a proof that a \emph{join tree} for $\querygraph$ can be maintained in \DynFO in such a way that after a single change to $\querygraph$ the maintained join tree only changes by a constant number of edges. 
We show that given a join tree for $\querygraph$ we can maintain the answer to \ahh{\schema} under changes of single edges of the join tree. The main result follows by compositionality properties of \DynFO.

We also give a hardness result for the case that also $\datagraph$ is subject to changes. If \ahh{\schema} is in \DynFO for every schema $\schema$ under changes of $\querygraph$ and $\datagraph$, then all \LOGCFL-problems are in (a variant of) \DynFO, which we believe not to be the case. So, this result is a strong indicator that maintenance under changes of $\datagraph$ is not possible in $\DynFO$.
Note that this result does not follow immediately from the fact that \ahhNA is \LOGCFL-complete: the \NL-complete problem of reachability in directed graphs is in \DynFO \cite{DattaKMSZ18} as well as a $\PTIME$-complete problem \cite{patnaikI97}, and these results do not imply that all \NL-problems and even all \PTIME-problems are in \DynFO, as this class is not known to be closed under the usual classes of reductions.

\paragraph*{Further related work.} %
In databases, Incremental View Maintenance is concerned with updating the result of a database query after a change of the input, see \cite{10.5555/310709.310737} for an overview. Koch \cite{Koch10incr} shows that a set of queries that include conjunctive queries can be maintained incrementally by low-complexity updates.
A system for maintaining the result of Datalog-like queries under changes of the data and the queries is described in \cite{GreenOW15live}.

\paragraph*{Organisation.}
We introduce preliminaries and the \DynFO framework in Section~\ref{section:preliminaries}. Section~\ref{section:upperbound} contains the maintenance result for \ahhNA  under changes of $\querygraph$, the hardness result for changes of $\datagraph$ is presented in Section~\ref{section:lowerbound}.
We conclude in Section~\ref{section:conclusion}.
This paper accompanies \cite{conferenceversion} and contains more proof details.

    \section{Preliminaries and Setting}\label{section:preliminaries}
We introduce some concepts and notation that we need throughout the paper. See also \cite{SchwentickVZ20} for an overview of Dynamic Complexity. We assume familiarity with first-order logic $\FO$, and refer to \cite{Libkin04} for basics of Finite Model Theory.

A \emph{(purely relational) schema} $\schema$ consists of a
finite set of relation symbols with a corresponding arity.
A \emph{structure} $\db$ over schema $\schema$ with finite domain $\domain$ has, for every $k$-ary relation symbol 
$R \in \schema$, a relation $R^\db \subseteq \domain^k$.
We assume that all structures come with a linear order $\leq$ on their domain $\domain$, which allows us to identify $\domain$ with $\{1, \ldots, n\}$, for $n = |\domain|$. We also assume that first-order formulas have access to this linear order and to compatible relations $+$ and $\times$ encoding addition and multiplication on $\{1, \ldots, n\}$.

\paragraph*{The dynamic complexity framework.}

In the dynamic complexity framework as introduced by Patnaik and Immerman \cite{patnaikI97},
the goal of a \emph{dynamic program} is to answer a standing query to an input structure $\inp$ under changes.
To do so, the program stores and updates an \emph{auxiliary structure} $\aux$, which is over the same domain as $\inp$. This structure consists of a set of \emph{auxiliary relations}.

The set of admissible changes to the input structure is specified by a set $\Delta$ of \emph{change operations}. We mostly consider the change operations $\ins_R$ and $\del_R$ for a relation $R$ of the input structure. A \emph{change} $\delta(\tpl a)$ consists of a change operation $\delta \in \Delta$ and a tuple $\tpl a$ over the domain of $\inp$. The change $\ins_R(\tpl a)$ inserts the tuple~$\tpl a$ into the relation $R$ and the change $\del_R(\tpl a)$ deletes $\tpl a$ from $R$.

For every change operation $\delta \in \Delta$ and every auxiliary relation $S$, a dynamic program has a first-order \emph{update rule} that specifies how $S$ is updated after a change over $\delta$. Such a rule is of the form \changeRule{$\delta(\tpl p)$}{$S(\tpl x)$}{$\varphi^S_\delta(\tpl p; \tpl x)$}, where the \emph{update formula} $\varphi^S_\delta$ is a first-order formula over the combined schema of $\inp$ and $\aux$.
After a change $\delta(\tpl a)$ is applied, the relation $S$ is updated to $\{\tpl b \mid (\inp, \aux) \models \varphi^S_\delta(\tpl a; \tpl b)\}$.

We say that a dynamic program $\prog$ \emph{maintains} a query $\query$ under changes from~$\Delta$ if a dedicated auxiliary relation $\ans$ contains the answer to $\query$ for the current input structure after each sequence of changes over $\Delta$.
The class $\DynFO$ contains all queries that can be maintained by dynamic programs with first-order update rules, starting from initially empty input and auxiliary relations. 
We also say that $\query$ \emph{can be maintained} in \DynFO under $\Delta$ changes.

In this paper we are interested in scenarios where only parts of the input are subject to changes. To have a meaningful setting we then have to allow non-empty initial input relations. We then say that a query can be maintained in $\DynFO$ \emph{starting from} non-empty inputs.
Sometimes we then also allow the auxiliary relations to be initialised within some complexity bound. We say that a query $\query$ is in \DynFO \emph{with $\calC$ initialisation}, for a complexity class $\calC$, if there is a $\calC$-algorithm $A$ such that $\query$ can be maintained in \DynFO if for an initial input~$\inp_0$ the initial auxiliary relations are set to the result of $A$ applied to $\inp_0$. 

The reductions usually used in dynamic complexity are bounded first-order reductions \cite{patnaikI97}. A reduction $f$ is \emph{bounded} if there is a global constant $c$ such that if the structure $\db'$ is obtained from the structure $\db$ by inserting or deleting one tuple, then $f(\db')$ can be obtained from $f(\db)$ by inserting and/or deleting at most $c$ tuples.
We will not directly employ these reductions here, but we will use the simple proof idea to show that \DynFO is closed under these reductions (see \cite{patnaikI97}): if a query $\query$ can be maintained by a dynamic program $\prog$ under insertions and deletions of single tuples, then there is also a dynamic program that can maintain $\query$ under insertions and deletions of up to $c$ tuples, for any constant~$c$. 
That dynamic program can be obtained by nesting $c$ copies of the update formulas of $\prog$.

\paragraph*{Hypergraphs and Homomorphisms.}

We use the term \emph{hypergraph} in a very broad sense. For this paper, a hypergraph $\calH$ is just a relational structure over a purely relational schema $\schema = \{E_1, \ldots, E_m\}$, that is, a structure $\calH = (\calV, E_1, \ldots, E_m)$, where the domain $\calV$ is a set of nodes and the relations $E_1, \ldots, E_m$ are sets of (labelled) hyperedges. This definition implies that the maximal size of any hyperedge, that is, the maximal arity of a relation $E_i$, is a constant that only depends on $\schema$.
Sometimes we ignore the labels and denote $\calH$ as a tuple $(\calV, \calE)$, where $\calE = E_1 \cup \cdots \cup E_m$ is the set of all hyperedges.

A \emph{spanning forest} of an undirected graph $G=(V,E)$ is defined in the usual way. We encode a spanning forest as a structure $(V, F, P)$ where $F$ is the set of spanning edges and $P$ is a ternary relation that describes paths in the spanning forest. A tuple $(s,t,u) \in P$ indicates that (1) $s$ and $t$ are in the same connected component of the spanning forest and (2) the unique path from $s$ to $t$ in the spanning forest is via the node $u$.
Patnaik and Immerman \cite{patnaikI97} have shown that spanning forests with this encoding can be maintained in \DynFO under insertions and deletions of single edges \cite[Theorem 4.1]{patnaikI97}.

A \emph{join forest} $J(\calH)$ of a hypergraph $\calH = (\calV, E_1, \ldots, E_m)$ is a forest whose nodes are the hyperedges of $\calH$, such that if two hyperedges $e, e'$ have a node $v \in \calV$ in common, then they are in the same connected component of $J(\calH)$ and all nodes on the unique path from $e$ to $e'$ in $J(\calH)$ are hyperedges of $\calH$ that also include $v$.
We encode a join forest using relations $F_{ij}$ and $P_{ijk}$ with the same intended meaning as for spanning forests, where $i,j,k \in \{1, \ldots, m\}$. The arity of $F_{ij}$ is the sum of the arities of $E_i$ and $E_j$, a tuple $(e,e') \in F_{ij}$ indicates that $J(\calH)$ has an edge between the hyperedges $e \in E_i$ and $e' \in E_j$. The use of $P_{ijk}$ is analogous.

We define that a hypergraph is \emph{acyclic} if it has a join forest. This definition coincides with the notion of \emph{$\alpha$-acyclicity} introduced by Fagin \cite{fagin1983degrees}. See also \cite[Section 2.2]{gottlob2001complexity} for a detailed discussion of this notion.

A \emph{homomorphism} from a hypergraph $\calH = (\calV, E_1^\calH, \ldots, E_m^\calH)$ to a hypergraph $\calH' = (\calV', E_1^{\calH'}, \ldots, E_m^{\calH'})$ is a map $h \colon \calV \to \calV'$ that preserves the hyperedge relations. So, for all relations $E_i$ and all tuples $(v_1, \ldots, v_\ell)$ over $\calV$, where $\ell$ is the arity of $E_i$, if $(v_1,\ldots, v_\ell) \in E_i^\calH$ is a hyperedge of $\calH$, then $(h(v_1),\ldots, h(v_\ell)) \in E_i^{\calH'}$ is a hyperedge of $\calH'$.

The main problem we study is the Acyclic Hypergraph Homomorphism problem \ahh{\schema}, where $\schema$ is a fixed schema. It asks, for two given hypergraphs 
$\querygraph$ and $\datagraph$ over schema $\schema$ (where
$\querygraph$ is acyclic), also called \emph{query hypergraph} and \emph{data hypergraph} respectively, whether there is a homomorphism from $\querygraph$ to $\datagraph$.

    \section{Maintenance under Changes of the Query Hypergraph}\label{section:upperbound}

The goal of this section is to show that \ahhNA can be maintained under changes of the query hypergraph $\querygraph$, as long as it stays acyclic. We also show that a \DynFO program can recognise that a change would make $\querygraph$ cyclic. So, we do not need to assume that only changes occur that preserve acyclicity, if we allow the program to ``deny'' all other changes.

We introduce some notation of \cite{gottlob2001complexity}.
The \emph{weighted hyperedge graph} $\WG(\calH)$ of a hypergraph $\calH$ is the undirected weighted graph $\WG(\calH) = (V,E,w)$ whose nodes $V$ are the hyperedges of $\calH$ and the set $E$ contains an undirected edge $(e,e')$ if $e, e'$ are different hyperedges of $\calH$ that have at least one node in common. The weight $w((e, e'))$ of such an edge is the number of nodes that $e$ and $e'$ have in common.

The \emph{weight} $\weight(\calH)$ of a hypergraph $\calH$ is the sum over the degrees of the non-isolated nodes of $\calH$, where each degree is decremented by one. So, if for $\calH = (\calV, E_1, \ldots, E_m)$ the set $\calV_{\mtext{ni}} \subseteq \calV$ contains all nodes of $\calH$ that appear in at least one hyperedge, then $\weight(\calH) = \sum_{v \in \calV_{\mtext{ni}}} (\deg(v) -1)$.

The following lemma provides the basis for our approach. It was originally proven in \cite{bernstein1981power}, we follow the presentation of \cite[Proposition 3.5]{gottlob2001complexity}.

\begin{lemma}[{\cite{bernstein1981power}, see also \cite{gottlob2001complexity}}]
\label{lem:acy_equiv}
Let $\calH$ be a hypergraph. 
\begin{enumerate}[label=(\alph*), ref=(\alph*)]
\item
The hypergraph $\calH$ is acyclic if and only if the weight $\weight(\calH)$ of $\calH$ is equal to the weight $\weight(\msfWG(\WG(\calH)))$ of a maximal-weight spanning forest of $\WG(\calH)$.
\item
If $\calH$ is acyclic, then $\msfWG(\WG(\calH))$ is a join forest of $\calH$.
\end{enumerate}
\end{lemma}

Using this lemma, we prove that a dynamic program can maintain acyclicity of hypergraphs, as well as a join forest that only changes moderately when the input hypergraph is changed.

\begin{theorem}\label{theorem:jf}
Let $\schema = \{E_1, \ldots, E_m\}$ be a fixed schema. The following can be maintained in \DynFO under insertions and deletions of single hyperedges:
\begin{enumerate}[label=(\alph*),ref=(\alph*)]
 \item whether a hypergraph over $\schema$ is acyclic, and
 \item a join forest for an acyclic hypergraph $\calH$ over $\schema$, as long as $\calH$ stays acyclic. Moreover, there is a global constant $c_\schema$ such that if $J(\calH)$ is the maintained join forest for $\calH$ and $J(\calH)'$ is the maintained join forest after a single hyperedge is inserted or deleted, then $J(\calH)$ and $J(\calH)'$ differ by at most $c_\schema$ edges.
\end{enumerate}
\end{theorem}

The proof follows the idea that is brought forth by Lemma \ref{lem:acy_equiv}: we show that a maximal-weight spanning forest of $\WG(\calH)$ and its weight can be maintained. This weight is compared with the weight of $\calH$, which is easy to maintain. If the weights are equal, then $\calH$ is acyclic and the spanning forest is a join forest.

Already Patnaik and Immerman \cite{patnaikI97} describe how a spanning forest of an undirected graph can be maintained under changes of single edges, and their procedure \cite[Theorem 4.1]{patnaikI97} can easily be extended towards maximal-weight spanning forests. 
However, we face the problem that inserting and deleting hyperedges of $\calH$ implies insertions and deletions of nodes of $\WG(\calH)$. While a spanning forest can easily be maintained in \DynFO under node insertions, it is an open problem to maintain a spanning forest under node deletions: if the spanning forest is a star and its center node is deleted, then it seems that a spanning forest of the remaining graph needs to be defined from scratch, which is not possible using \FO formulas.
We circumvent this problem and show that we can maintain a spanning forest where the degree of every node is bounded by a constant.

\begin{proof}
We show how a maximal-weight spanning forest of $\WG(\calH)$ and the weight $\weight(\calH)$ can be maintained; the result then follows using Lemma~\ref{lem:acy_equiv}.

We start with the weight $\weight(\calH)$. If a hyperedge $e$ is inserted, then the weight of the hypergraph increases by the number of nodes it contains that were not isolated before the insertion. Similarly, if $e$ is deleted, then the weight decreases by the number of nodes it contains that do not become isolated. This update can easily be expressed by first-order formulas.

Now we consider maintaining a maximal-weight spanning forest of $\WG(\calH)$.

Let \maxarity be the maximal arity of a relation in $\schema$. Any hyperedge of $\calH = (\calV, E_1, \ldots, E_m)$ can only include at most \maxarity many nodes and there are at most $r \df 2^\maxarity-1$ many different non-empty sets of nodes that a fixed hyperedge can have in common with any other hyperedge.
We show how to maintain a maximal-weight spanning forest of 
$\WG(\calH)$ where each node has degree at most $2r$. More specifically, for any node $e$ of $\WG(\calH)$ (which is a hyperedge of $\calH$) and each non-empty set $A$ of nodes appearing in $e$, the maintained spanning forest contains at most two edges $(e,e_1), (e,e_2)$ such that the set of nodes that $e$ has in common with $e_1$ and $e_2$, respectively, is exactly $A$. We call this property \emph{Invariant $(\star)$}.

We assume that our auxiliary relations contain a maximal-weight spanning forest $S(\calH)$ and its weight, and that $S(\calH)$ satisfies Invariant $(\star)$. 
This is trivially satisfied by an empty spanning forest for an initially empty hypergraph. 
We show how the invariant can be satisfied again after a change. 

In the following, we say that $e'$ is an \emph{$A$-neighbour} of $e$ if these hyperedges have exactly the nodes $A$ in common. The number of $A$-neighbours of $e$ (that $e$ has an edge to in $S(\calH)$) is called its \emph{$A$-degree} (with respect to $S(\calH)$).

\paragraph*{Insertion of a hyperedge $e$.}
Suppose a hyperedge $e$ is inserted into the hypergraph $\calH$, resulting in the hypergraph $\calH'$. Let $B$ be the set of nodes that occur in $e$. For each non-empty $A \subseteq B$, let $E_A$ be the edges of $e$ in $\WG(\calH')$ to its $A$-neighbours.
We adapt $S(\calH)$ in stages, one stage per subset $A$, and in each stage the spanning forest is changed by at most two edges. As the number of stages is bounded by the constant $r$, the maintained spanning forests before and after the update only differ by a constant number of edges.

Let $A_1, \ldots, A_\ell$ be a sequence of all non-empty 
subsets of $B$, partially ordered by their size, starting with 
the largest. So, $A_1 = B$.
Stage $i$ for an arbitrary $1 \leq i \leq \ell$ works as
follows. 
Suppose that $S_{i-1}$ is a maximal spanning forest of the graph that results from
$\WG(\calH)$ by adding the node $e$ and the edge set $\bigcup_{j \leq i-1} E_{A_j}$ and that satisfies 
Invariant $(\star)$ (so $S_0$ is a maximum spanning
forest for $\WG(\calH) \cup (\{e\}, \emptyset)$ and therefore also for $\WG(\calH)$).
Let $N$ be the $A_i$-neighbours of $e$. The hyperedges in $N$ 
form a clique in $\WG(\calH)$, as they all have at least the 
nodes $A_i$ in common, so they are in the same connected 
component $C$ of $S_{i-1}$. We consider two cases.

First, if $e$ is not in $C$, then let $e' \in N$ be some 
hyperedge that has $A_i$-degree at most $1$ with respect to 
$S_{i-1}$. Such an $e'$ needs to exist as $S_{i-1}$ is a 
forest. Then $S_i \df S_{i-1} \cup \{(e,e')\}$ clearly is a 
spanning forest of $\WG(\calH) \cup \bigcup_{j \leq i}
E_{A_j}$. It is maximal, as replacing some spanning edge 
$(e_1,e_2)$ by another edge from $e$ to an $A_i$-neighbour
cannot increase the weight: if this would be the case for some 
edge $(e,e'')$, then $S_{i-1}$ cannot be maximal, because the 
edge $(e',e'')$ has at least the same weight as the edge 
$(e,e'')$ and replacing $(e_1,e_2)$ by $(e',e'')$ would 
therefore create a spanning forest with larger weight than 
$S_{i-1}$. Invariant $(\star)$ is also satisfied by $S_i$.

Second, if $e$ is in $C$, then let $(e_1,e_2)$ be the minimal-weight edge in $S_{i-1}$ on the path from $e$ to any hyperedge in $N$. If the weight of this edge is at least $|A_i|$, then $S_i \df S_{i-1}$, as the weight of the spanning forest cannot be increased by incorporating an edge from $E_{A_i}$. Otherwise, $S_i$ results from $S_{i-1}$ by removing the edge $(e_1,e_2)$ and adding an edge from $e$ to one of its $A_i$-neighbours with $A_i$-degree at most $1$ with respect to $S_{i-1}$. With the same arguments as in the other case, $S_i$ is a maximal-weight spanning forest of $\WG(\calH) \cup \bigcup_{j \leq i} E_{A_j}$, it also satisfies Invariant $(\star)$.

The updates of the spanning forest relation $F_{ij}$ can be expressed by first-order formulas using the relations $P_{ijk}$. These relations can be updated as in the proof of \cite[Theorem 4.1]{patnaikI97}, as the relations $F_{ij}$ change only by a constant number of tuples. The weight of the spanning forest can also be updated easily.

\paragraph*{Deletion of a hyperedge $e$.}
Suppose $e$ is deleted from $\calH$, resulting in the hypergraph $\calH'$. As the maintained spanning forest $S(\calH)$ satisfies Invariant $(\star)$, the degree of $e$ in $S(\calH)$ is bounded by the constant $r$.
Therefore, the procedure of \cite[Theorem 4.1]{patnaikI97} only needs to be applied for a constant number of edge deletions.
If by a deletion of a spanning tree edge a connected component of $S(\calH)$ decomposes into two components $C_1$ and $C_2$, then we need to ensure that a potentially selected replacement edge results in a spanning forest that satisfies Invariant $(\star)$ again.
So, suppose that the components $C_1$ and $C_2$ are connected in $\WG(\calH')$, and let $(e_1,e_2)$ be a maximal-weight edge that connects them. Let $A$ be the set of nodes that $e_1$ and $e_2$ have in common. Without loss of generality, we suppose that $e_1$ and $e_2$ have $A$-degree at most $1$ with respect to $S(\calH)$. 
If this is not the case for example for $e_1$, then there needs to be a $A$-neighbour $e'_1$ in $C_1$ that has $A$-degree at most $1$ with respect to $S(\calH)$, which then can be used instead of $e_1$.
Adding $(e_1, e_2)$ to the remaining spanning forest will therefore result in a maximal-weight spanning forest for the changed hypergraph which satisfies Invariant $(\star)$. Also, the maintained spanning forest differs only by at most $2r$ edges from its previous version.
The weight of the spanning forest can easily be updated as well.
\end{proof}

We now present the main maintenance result of this paper.

\begin{theorem}\label{theorem:upperbound}
Let $\schema = \{E_1, \ldots, E_m\}$ be a fixed schema. The problem \ahh{\schema} can be maintained in \DynFO, starting from an arbitrary initial hypergraph $\datagraph$ and an initially empty hypergraph $\querygraph$, under insertions and deletions of single hyperedges of $\querygraph$, as long as this hypergraph stays acyclic.
\end{theorem}

The proof uses the idea of Yannakakis' algorithm \cite{yannakakis1981algorithms} for evaluating a conjunctive query. This algorithm processes a join tree for a query $\querygraph$ in a bottom-up fashion. In a first step, for each node $E_i(\tpl x)$ of the join tree (which is a hyperedge of $\querygraph$) all assignments $\tpl y$ for its variables are stored such that $E_i(\tpl y)$ exists in the data hypergraph $\datagraph$.
Then, bottom-up, each inner node removes all of its variable assignments that are not consistent with the assignments of its children. So, an assignment $\tpl y$ for a node $E_i(\tpl x)$ is removed if there is a child $E_j(\tpl x')$ of $E_i(\tpl x)$ such that no stored assignment $\tpl y'$ of that child agrees with $\tpl y$ on the common variables of $\tpl x$ and $\tpl x'$. All remaining stored assignments for $E_i(\tpl x)$ can be extended to a homomorphism for the subhypergraph of $\querygraph$ that consists of the hyperedges that are in the subtree of the join tree rooted at $E_i(\tpl x)$.
A homomorphism from $\querygraph$ to $\datagraph$ exists if after the join tree is processed the root has remaining assignments.

\begin{proof}
Let $\querygraph$ be an acyclic hypergraph over some schema $\schema$ and let $\datagraph$ be a hypergraph over the same schema. Also, let $J(\querygraph)$ be a join forest of $\querygraph$.

We adapt a technique that was used by Gelade, Marquardt and Schwentick \cite{GeladeMS12} to show that regular tree languages can be maintained in a subclass of \DynFO. 
For each triple $E_i, E_j, E_k$ of symbols from $\schema$ we maintain an auxiliary relation $H_{ijk}(\tpl r, \tpl x_1, \tpl x_2, \tpl y_1, \tpl y_2)$ with the following intended meaning.
A tuple $(\tpl r, \tpl x_1, \tpl x_2, \tpl y_1, \tpl y_2)$ is in $H_{ijk}$ if 
\begin{enumerate}[label=(\arabic*),ref=(\arabic*)]
 \item the hyperedges $E_i(\tpl r)$, $E_j(\tpl x_1)$ and $E_k(\tpl x_2)$ are present in $\querygraph$ and in the same connected component $C$ of $J(\querygraph)$, 
 \item when we consider $E_i(\tpl r)$ to be the root of $C$ then $E_j(\tpl x_1)$ is a descendant of $E_i(\tpl r)$ and $E_k(\tpl x_2)$ is a descendant of $E_j(\tpl x_1)$, and
 \item if we assume that there is a homomorphism $h_2$ of the subtree of $C$ rooted at $E_k(\tpl x_2)$ into $\datagraph$ such that $h_2(\tpl x_2) = \tpl y_2$, then it follows that there also is a homomorphism $h_1$ of the subtree of $C$ rooted at $E_j(\tpl x_1)$ into $\datagraph$ such that $h_1(\tpl x_1) = \tpl y_1$.
\end{enumerate}

Phrased differently, $(\tpl r, \tpl x_1, \tpl x_2, \tpl y_1, \tpl y_2) \in H_{ijk}$ means that the hyperedges in $J(\querygraph)$ which, considering $E_i(\tpl r)$ to be the root, are in the subtree of $E_j(\tpl x_1)$ but not in the subtree of $E_k(\tpl x_2)$, can be mapped into $\datagraph$ by a homomorphism that maps the elements $\tpl x_1$ to $\tpl y_1$ and the elements $\tpl x_2$ to $\tpl y_2$.

If $(\tpl r, \tpl x_1, \tpl x_2, \tpl y_1, \tpl y_2) \in H_{ijk}$ holds we say that $\tpl y_1$ is a \emph{valid partial assignment} for $E_j(\tpl x_1)$ \emph{down to} $(E_k(\tpl x_2), \tpl y_2)$.

Notice that from these relations one can first-order define relations $H'_{ij}(\tpl r, \tpl x, \tpl y)$ with the intended meaning that $(\tpl r, \tpl x, \tpl y) \in H'_{ij}$ if
\begin{enumerate}[label=(\arabic*),ref=(\arabic*)]
\item the hyperedges $E_i(\tpl r)$ and $E_j(\tpl x)$ are in the same connected component $C$ of $J(\querygraph)$, and
\item when we consider $E_i(\tpl r)$ to be the root of $C$ then there is a homomorphism $h$ of the subtree of $C$ rooted at $E_j(\tpl x)$ into $\datagraph$ such that  $h(\tpl x) = \tpl y$. 
\end{enumerate}
For this, a first-order formula existentially quantifies a hyperedge $E_k(\tpl x_2)$ and a tuple $\tpl y_2$ of elements, and checks that $E_k(\tpl x_2)$ is a leaf of the component $C$ with root $E_i(\tpl r)$, that the hyperedge $E_k(\tpl y_2)$ exists in $\datagraph$ and that $(\tpl r, \tpl x, \tpl x_2, \tpl y, \tpl y_2) \in H_{ijk}$ holds.
Whether a node is a leaf in a join tree can be expressed using the join tree's paths relations $P_{ijk}$, all other conditions are clearly first-order expressible.
We assume in the following that these relations are available. If $(\tpl r, \tpl x, \tpl y) \in H'_{ij}$ holds we say that $\tpl y$ is a \emph{valid partial assignment} for $E_j(\tpl x)$.

We argue next that if we can maintain these auxiliary relations under insertions and deletions of single edges of the join forest, then the statement of the theorem follows.

Notice that from the auxiliary relations a first-order formula can express whether a homomorphism from $\querygraph$ to $\datagraph$ exists. To check this, a formula needs to express that for every connected component of $J(\querygraph)$ there is a homomorphism from this component to $\datagraph$. 
This is the case if for each hyperedge $E_i(\tpl r)$ of $\querygraph$ there is a tuple $\tpl y$ such that $(\tpl r,\tpl r,\tpl y)$ is in $H'_{ii}$.

It remains to argue that it is sufficient to maintain the auxiliary relations under changes of single edges of the join forest.
From Theorem \ref{theorem:jf} we know that a join forest $J(\querygraph)$ for $\querygraph$ can be maintained in \DynFO under insertions and deletions of single hyperedges, as long as it stays acyclic. Moreover, after each edge change, the maintained join forest only differs in a constant number of edges from its previous version.
If we have a dynamic program that is able to process single edge changes of the join forest, then by nesting its update formulas $c$ times we can obtain a dynamic program $\prog'$ that is able to process $c$ edge changes at once.
In summary, a dynamic program $\prog$ for \ahhNA maintains a join forest as described by Theorem \ref{theorem:jf} and after every change of a hyperedge it uses $\prog'$ to update the auxiliary relations and to decide whether a homomorphism exists.

Now we explain how the relations $H_{ijk}$ can be maintained by first-order formulas under insertions and deletions of single edges of the join forest. For notational simplicity we assume that the schema $\schema$ of $\querygraph$ consists of a single relation $E$. It follows that we only have one auxiliary relation $H$ that needs to be maintained.

\paragraph*{Edge insertions.}

When an edge $(e_1,e_2)$ is inserted into the join forest, the two connected components $C_1$ of $E(e_1)$ and $C_2$ of $E(e_2)$ get connected. The auxiliary relations for all other connected components remain unchanged.
We explain under which conditions a tuple $\tpl t = (\tpl r, \tpl x_1, \tpl x_2, \tpl y_1, \tpl y_2)$ is contained in the updated version of $H$, where we assume that $E(\tpl r)$ is from $C_1$. For hyperedges from $C_2$ the reasoning is symmetric. We assume that $E(\tpl x_1)$ is a descendant of $E(\tpl r)$ and $E(\tpl x_2)$ is a descendant of $E(\tpl x_1)$ in the combined connected component rooted at $E(\tpl r)$; if this is not the case, $\tpl t$ is not in the updated version of $H$.

We distinguish three cases. First, assume that $E(\tpl x_1)$ and $E(\tpl x_2)$ are in $C_1$. 
If $E(e_1)$ is not in the subtree of $E(\tpl x_1)$ or is in the subtree of $E(\tpl x_2)$, then no change regarding $\tpl t \in H$ is necessary. 
 
Otherwise, let $E(\tpl x_{\mtext{lca}})$ be the lowest common ancestor of $E(\tpl x_2)$ and $E(e_1)$ in the join tree with root $E(\tpl r)$ and let $E(\tpl x_c^1), \ldots, E(\tpl x_c^m)$ be the children of $E(\tpl x_{\mtext{lca}})$, where $E(\tpl x_c^1)$ is the ancestor of $E(\tpl x_2)$ and $E(\tpl x_c^m)$ is the ancestor of $E(e_1)$.
With the help of the old version of $H$ first-order formulas can determine the valid partial assignments for all children $E(\tpl x_c^i)$ for $i \geq 2$ and the valid partial assignments for $E(\tpl x_c^1)$ down to $(E(\tpl x_2), \tpl y_2)$. This is immediate for all $E(\tpl x_c^i)$ with $i < m$, we now explain it for $E(\tpl x_c^m)$.

To check whether a tuple $\tpl y_c^m$ is a valid partial assignment for $E(\tpl x_c^m)$, a first-order formula can first determine the valid partial assignments for $E(e_2)$ for the component $C_2$ with root $E(e_2)$, which are given by $H'$. With this information it can check which valid partial assignments for $E(e_1)$ for the component $C_1$ with root $E(\tpl r)$ are still valid for the union of $C_1$ and $C_2$ with root $E(\tpl r)$. To do so, it checks for a (formerly) valid assignment $E(e_1)$ whether there is a valid assignment for $E(e_2)$ such that they agree on the shared elements.
The tuple $\tpl y_c^m$ is a valid partial assignment for $E(\tpl x_c^m)$ if it is a valid partial assignment for $E(\tpl x_c^m)$ down to $(E(e_1),\tpl y_{e_1})$, for some valid partial assignment $\tpl y_{e_1}$ for $E(e_1)$.

With the information on the children, a first-order formula 
can determine the valid partial assignments for 
$E(\tpl x_{\mtext{lca}})$ down to $(E(\tpl x_2), \tpl y_2)$. 
This only involves a check whether for a candidate assignment 
$\tpl y_\mtext{lca}$ a corresponding hyperedge 
$E(\tpl y_\mtext{lca})$ exists in $\datagraph$ and 
whether every child $E(\tpl x_c^i)$ has a valid partial 
assignment that agrees with $\tpl y_\mtext{lca}$ on the 
elements that are shared by $\tpl x_c^i$ and 
$\tpl x_\mtext{lca}$.

The tuple $\tpl t$ is in the updated version of $H$ if and 
only if $(\tpl r, \tpl x_1, \tpl x_\mtext{lca}, \tpl y_1,  
\tpl y_\mtext{lca})$ is in the old version of $H$, for some 
valid partial assignment $\tpl y_\mtext{lca}$ of $E(\tpl 
x_{\mtext{lca}})$ down to $(E(\tpl x_2), \tpl y_2)$,

As a second case, assume that $E(\tpl x_1)$ is in $C_1$ and $E(\tpl x_2)$ is in $C_2$. This case is very similar to the case we just discussed and we do not spell out the details.

As a last case, assume that  $E(\tpl x_1)$ and $E(\tpl x_2)$ 
are both in $C_2$. This case is very simple, as $t \in H$ holds 
after the update precisely if $(\tpl e_2, \tpl x_1,  
\tpl x_2, \tpl y_1, \tpl y_2) \in H$ holds before the update.

In all cases, the stated conditions can be expressed by first-order formulas. This is because the schema $\schema$ is fixed and therefore the arity of $E$ is constant, it follows that formulas can quantify over hyperedges and assignments. Also, formulas can determine whether a node is in a subtree of another node and the lowest common ancestor of two nodes using the paths relation $P_{ijk}$ of the join forest.

\paragraph*{Edge deletions.}

When an edge $(e_1,e_2)$ is deleted from the join forest, one 
connected component is split into the two connected components 
$C_1$ of $E(e_1)$ and $C_2$ of $E(e_2)$. Again, the auxiliary 
relations for all other connected components remain unchanged.
As for the insertion case, we explain under which conditions a 
tuple $\tpl t = (\tpl r, \tpl x_1, \tpl x_2, \tpl y_1, 
\tpl y_2)$ is contained in the updated version of $H$, for 
a root $E(\tpl r)$ from $C_1$. We assume that $E(\tpl x_1)$ is 
a descendant of $E(\tpl r)$ and $E(\tpl x_2)$ is a descendant 
of $E(\tpl x_1)$ in the component $C_1$ rooted at $E(\tpl r)$; 
otherwise, $\tpl t$ is not in the updated version of $H$.

If $E(e_1)$ is not in the subtree of $E(\tpl x_1)$ or is in 
the subtree of $E(\tpl x_2)$, then no change regarding
$\tpl t \in H$ is necessary. 
Otherwise the update is performed very similarly to the 
corresponding insertion case detailed above. 
The only difference is the way the valid partial assignments 
for $E(e_1)$ are determined.
Notice that a first-order formula can determine the valid 
partial assignments for all (remaining) children of $E(e_1)$, 
as they are given by the relation $H'$. A tuple $\tpl y_{e_1}$ 
is a valid partial assignment for $E(e_1)$ if the hyperedge 
$E(\tpl y_{e_1})$ exists in $\datagraph$ and if all children 
have a valid partial assignment that agrees with 
$\tpl y_{e_1}$ on the shared elements.
\end{proof}

     \section{Hardness under Changes of the Data Hypergraph}\label{section:lowerbound}
We have seen in the previous section that one can maintain the existence of homomorphisms in \DynFO if only the query hypergraph $\querygraph$ may change and the data hypergraph $\datagraph$ remains the same. 
The dynamic program we constructed for the proof of Theorem \ref{theorem:upperbound} can not directly cope with changes of $\datagraph$. This is because $\querygraph$ might contain several hyperedges $E_i(\tpl x_1), \ldots, E_i(\tpl x_m)$ over a single relation~$E_i$: if a change of $\datagraph$ occurs, then the number of nodes in the join tree for which we have to take this change into account when updating partial valid assignments is a priori unbounded. 
If we disallow multiple hyperedges over the same relation in
 $\querygraph$, then we can actually allow a change to replace an arbitrary number of $\datagraph$-hyperedges, as long as each change only affects a single
relation of $\datagraph$. Such a restriction of $\querygraph$ translates to \emph{self-join free} acyclic conjunctive queries.

\begin{corollary}\label{corollay:self-join-free}
Let $\schema = \{E_1, \ldots, E_m\}$ be a fixed schema. As long as $\querygraph$ remains acyclic and contains at most one hyperedge $E_i(\tpl x)$ for each relation $E_i \in \schema$, the problem \ahh{\schema} can be maintained in \DynFO under insertions and deletions of single hyperedges of $\querygraph$ and under arbitrary changes of a single relation of $\datagraph$. 
\end{corollary}
\begin{proofsketch}
To adapt the proof of Theorem \ref{theorem:upperbound}, it suffices to show how after changing some relation $E_i$ of $\datagraph$ one can determine the valid partial assignments for the single node $E_i(\tpl x)$ in the join forest $J(\querygraph)$. As the valid partial assignments for its children did not change, this only involves to check for each tuple $\tpl y$ such that the hyperedge $E_i(\tpl y)$ exists in $\datagraph$ whether each child in $J(\querygraph)$ has a valid partial assignment that agrees with $\tpl y$ on all elements it has in common with $\tpl x$.
\end{proofsketch}

In the remainder of this section, we will see that if $\querygraph$ might be an arbitrary acyclic hypergraph, then a maintenance result for \ahh{\schema} under changes of $\datagraph$ is unlikely, even if in turn $\querygraph$ is not allowed to change.

Gottlob et al. \cite{gottlob2001complexity} show that it is \LOGCFL-complete to decide whether from a given acyclic hypergraph $\querygraph$ there is a homomorphism into a hypergraph $\datagraph$. The complexity class \LOGCFL contains all problems that can be reduced in logarithmic space to a context-free language. This class is contained in $\ACo$, contains $\NL$ and is equivalent to logspace-uniform $\SACo$ \cite{Venkateswaran91}, the class of problems decidable by logspace-uniform families of semi-unbounded Boolean circuits of polynomial size and logarithmic depth. A semi-unbounded Boolean circuit consists of or-gates with unbounded fan-in and and-gates with fan-in $2$. There are no negation gates, but for each input gate $x_i$ there is an additional input gate $\neg x_i$ that carries the negated value of $x_i$.

In their article, Gottlob et al. \cite{gottlob2001complexity} show that there is a schema $\schema$ such that every $\SACo$ problem can be reduced in logarithmic space to \ahh{\schema}.
We slightly adapt their construction and show that the hardness result also holds for \emph{bounded} logspace reductions. Furthermore, if a reduction $f$ maps an instance $x$ to an instance $f(x)$, then the change to $f(x)$ induced by a change to $x$ is first-order definable.  %

\begin{theorem}[{adapted from \cite[Theorem 4.8]{gottlob2001complexity}}]\label{theorem:reduction}
\begin{enumerate}[label=(\alph*),ref=(\alph*)]
\item 
\label{item:AHH_LOGCFL_hardness}
There is a schema $\schema$ that contains at
most binary relations such that \ahh{\schema}
is hard for \LOGCFL under logspace reductions.
\item 
\label{item:reduction_props}
Let $L \in \LOGCFL$. There is a logspace reduction $f_L$ from $L$ to \ahh{\schema} that satisfies the following properties.
Assume that $x,x'$ are instances of $L$ with $|x| =|x'|$ and let $(\querygraph, \datagraph) = f_L(x)$
and $(\querygraph', \datagraph') = f_L(x')$. Then:
\begin{enumerate}[label=(\roman*),ref=(\roman*)]
   \item 
\label{item:q_eq_q_prime}   
   $\querygraph = \querygraph'$,
  \item 
  \label{item:db_changes}   
  if $x$ and $x'$ differ only in one bit, then 
  $\datagraph'$ differs from $\datagraph$ by at most $c$ 
  hyperedges, for a global constant $c$, and
  \item 
  \label{item:new_db_def}  
  $\datagraph'$ is first-order definable from 
  $\datagraph$, $x$ and $x'$.
  \end{enumerate}
\end{enumerate}
\end{theorem}

\begin{proof}
Let $L$ be a problem from $\LOGCFL$. As $\LOGCFL = \text{logspace-uniform } \SACo$, there is a logspace-uniform family $(C_n)_{n \in \N}$ of circuits that decides $L$, where a circuit $C_n$ has size at most $n^k$ for some $k \in \N$, logarithmic depth in $n$, and the fan-in of every and-gate is bounded by $2$. Without loss of generality, see \cite[Lemma 4.6]{gottlob2001complexity}, we can assume that $C_n$ also has the following \emph{normal form}:
\begin{enumerate}[label=(\arabic*),ref=(\arabic*)]
 \item the circuit consists of layers of gates, and the gates of layer $i$ receive all their inputs from gates at layer $i-1$,
 \item all layers either only contain or-gates or only contain and-gates,
 \item the first layer after the inputs consists of or-gates, 
 \item if layer $i$ is a layer of or-gates, then layer $i+1$ only consists of and-gates, and vice versa, and
 \item the output gate is an and-gate.
\end{enumerate}
A circuit of this form accepts its input if and only if a proof tree can be homomorphically mapped into it. A \emph{proof tree} $T_n$ for a circuit $C_n$ in normal form has the same depth as $C_n$ and an and-gate as its root. Each and-gate of the proof tree has two or-gates as its children, and each or-gate has one child, which is a gate labelled with the constant $1$ for an or-gate at the lowest layer, and an and-gate for all other or-gates. Note that a proof tree is acyclic.

If there is a homomorphism that maps each constant $1$ of the proof tree to an input gate of the circuit that is set to $1$, each and-gate of the proof tree to an and-gate of the circuit, and for each and-gate of the proof tree its two children to different or-gates in the circuit, then all gates of the circuit that are in the image of the homomorphism evaluate to $1$ for the current input. Therefore, the output gate also evaluates to $1$, and the circuit accepts its input.
It is also clear that if the circuit accepts its input, then there is a homomorphism from the proof tree into the circuit.

We encode circuits and proof trees over the schema $\schema = \{\mtext{0}, \mtext{1}, \mtext{or},\allowbreak \mtext{and-left},\allowbreak \mtext{and-right}\}$.
Each gate $g$ is encoded by a tuple $\mtext{enc}(g)$ of $k$ elements.
If $g$ is an and-gate with children $g_1, g_2$, then this is encoded by tuples $(\mtext{enc}(g),\mtext{enc}(g_1))\in \mtext{and-left}$ and $(\mtext{enc}(g),\mtext{enc}(g_2))\in \mtext{and-right}$. If $g$ is an or-gate and $g'$ is one of its children, then this is encoded by the tuple $(\mtext{enc}(g),\mtext{enc}(g'))\in \mtext{or}$.
The relations $\mtext{0}$ and $\mtext{1}$ are used to encode constants and assignments of input gates in the obvious way.

We use the two relations $\mtext{and-left}, \mtext{and-right}$ to ensure that a homomorphism from a proof tree to a circuit maps the two children of an and-gate to two different or-gates.

From the proof of \cite[Theorem 4.8]{gottlob2001complexity} it follows that from an input $x$ of $L$ with $|x| = n$ the corresponding circuit $C_n(x)$, which results from $C_n$ by assigning constants to its inputs gates as specified by $x$, and the corresponding proof tree $T_n$ can be computed in logarithmic space. In conclusion, this proves that the function $f_L$ that maps $x$ to $(T_n, C_n(x))$ is a logspace reduction from $L$ to \ahh{\schema}, and therefore that \ahh{\schema} is hard for \LOGCFL under logspace reductions.

We now proceed to prove part~\ref{item:reduction_props} of the theorem statement.
Consider two input instances $x, x'$ for $L$ with $|x| = |x'| = n$. Both $x$ and $x'$ are inputs of the circuit $C_n$, so the same proof tree is constructed for them by $f_L$, yielding part~\ref{item:reduction_props}\ref{item:q_eq_q_prime}. The only differences in the images of $f_L$ are the assignments of constants to the input gates of $C_n$. 
If $x$ and $x'$ only differ in one bit, say, the first bit that is represented by the input gate $g_1$, then we have $\mtext{enc}(g_1) \in \mtext{0}$ and $\mtext{enc}(\neg g_1) \in \mtext{1}$ for one input, and $\mtext{enc}(g_1) \in \mtext{1}$ and $\mtext{enc}(\neg g_1) \in \mtext{0}$ for the other input. So, the encodings of the circuit only differ by $4$ tuples, which implies part~\ref{item:reduction_props}\ref{item:db_changes}.
Towards part~\ref{item:reduction_props}\ref{item:new_db_def}, we can ensure that these tuples are first-order definable by using an appropriate encoding $\mtext{enc}$ of the gates, for example by encoding the $i$-th input gate by the $i$-th tuple in the lexicographic ordering of $k$-tuples over the domain.
\end{proof}

Building on the hardness result of Theorem \ref{theorem:reduction}, we can show that if \ahh{\schema} can be maintained in \DynFO under changes of $\datagraph$, then all \LOGCFL-problems are in \DynFO if we allow a \PTIME initialisation.
This would be a breakthrough result, as there are already problems in uniform $\ACz[2]$ (problems decidable by uniform circuits with polynomial size, constant depth and not-, and-, or- and \mbox{modulo 2}-gates with arbitrary fan-in), a much smaller complexity class, that we do not know how to maintain in \DynFO \cite{VortmeierZ20}.

\begin{theorem}\label{theorem:lowerbound}
If for arbitrary schema $\schema$ the problem \ahh{\schema} can be maintained in \DynFO under insertions and deletions of single hyperedges from $\querygraph$ and $\datagraph$, as long as $\querygraph$ stays acyclic, then every problem $L \in \LOGCFL$ can be maintained in \DynFO with \PTIME initialisation under insertions and deletions of single tuples.
\end{theorem}

The same even holds under the condition that \ahh{\schema} can only be maintained under changes of single hyperedges of $\datagraph$, but starting from an arbitrary initial acyclic hypergraph $\querygraph$, even if a \PTIME initialisation of the auxiliary relations is allowed. So, we can take this theorem as a strong indicator that \ahhNA might not be in \DynFO under changes of $\datagraph$.

\begin{proof}
Let $L \in \LOGCFL$ be arbitrary. Let $\schema$ be the schema and $f_L$ the reduction guaranteed to exist by Theorem \ref{theorem:reduction} such that $f_L$ is a reduction from $L$ to \ahh{\schema}.
Let $\prog$ be a dynamic program that maintains \ahh{\schema} under insertions and deletions of single hyperedges.
We construct a dynamic program $\prog'$ with \PTIME initialisation that maintains $L$.

For an initially empty input structure $\inp$ over a domain of size $n$, the initialisation first constructs the corresponding $\SACo$-circuit $C_n(\inp)$, with the input bits set as given by $\inp$, and the proof tree $T_n$ and stores them in auxiliary relations. This is possible in $\LOGSPACE \subseteq \PTIME$. Then, using polynomial time, it simulates $\prog$ for a sequence of insertions that lead to $C_n(\inp)$ and $T_n$ from initially empty hypergraphs and stores the produced auxiliary relations. 

When a change of $\inp$ occurs, $\prog'$ identifies the constantly many changes of $C_n(\inp)$ that are induced by the change, which is possible in first-order logic thanks to Theorem \ref{theorem:reduction}, and simulates $\prog$ for these changes.
\end{proof}

    \section{Conclusion and Further Work}\label{section:conclusion}
In this paper we studied under which conditions the problem \ahhNA can be maintained in \DynFO.
Our main result is that this problem is in \DynFO under changes of single hyperedges of the query 
hypergraph $\querygraph$, on the condition that it remains
acyclic. 
This result directly implies that the result of acyclic Boolean conjunction queries can be maintained in \DynFO. As the corresponding dynamic program, see proof of Theorem \ref{theorem:upperbound}, also maintains partial assignments of existing homomorphisms, this can straightforwardly be extended also to non-Boolean acyclic conjunctive queries.

We have also seen that it is unlikely that \ahhNA is in \DynFO under changes of the data hypergraph $\datagraph$.

In the static setting, the homomorphism problem is not only tractable for acyclic hypergraphs $\querygraph$, but for a larger class of graphs \cite{DalmauKV02} which includes the class of graphs with \emph{bounded treewidth}, see \cite{gottlob2001complexity}.
It is therefore interesting whether our \DynFO maintenance result can also be extended to allow for cyclic hypergraphs~$\querygraph$, in particular to allow hypergraphs of treewidth at most $k$, for some~$k$.
Results of this form would probably require an analogous result to Theorem \ref{theorem:jf}, so, that a tree decomposition of some width $f(k)$ can be maintained for every hypergraph with treewidth at most $k$, and that any change of the hypergraph leads to a maintained tree decomposition that can be obtained from its previous version by a constant number of changes.

Outside the \DynFO framework, maintenance of tree decompositions for graphs with treewidth $k=2$, that is, series-parallel graphs, is considered in \cite{Bodlaender93a}, but a change of the graph may affect a logarithmic number of nodes of the tree decomposition.
Preliminary unpublished results show (using different techniques than \cite{Bodlaender93a}) that for graphs with treewidth $2$ a tree decomposition can indeed be maintained in \DynFO. It is so far unclear whether tree decompositions can also be maintained in a way that only a constant-size part changes after a change of the graph.

\subsection*{Acknowledgements.} This project has received funding
from the European Union's Horizon 2020 research and innovation
programme under grant agreement No 682588.    

\bibliography{bibliography}

\begin{thebibliography}{10}
\providecommand{\url}[1]{\texttt{#1}}
\providecommand{\urlprefix}{URL }
\providecommand{\doi}[1]{https://doi.org/#1}

\bibitem{beeri1983desirability}
Beeri, C., Fagin, R., Maier, D., Yannakakis, M.: On the desirability of acyclic
  database schemes. Journal of the ACM (JACM)  \textbf{30}(3),  479--513
  (1983). \doi{10.1145/2402.322389}

\bibitem{bernstein1981power}
Bernstein, P.A., Goodman, N.: Power of natural semijoins. SIAM Journal on
  Computing  \textbf{10}(4),  751--771 (1981). \doi{10.1137/0210059}

\bibitem{Bodlaender93a}
Bodlaender, H.L.: Dynamic algorithms for graphs with treewidth 2. In: van
  Leeuwen, J. (ed.) Graph-Theoretic Concepts in Computer Science, 19th
  International Workshop, {WG} '93, Utrecht, The Netherlands, June 16-18, 1993,
  Proceedings. Lecture Notes in Computer Science, vol.~790, pp. 112--124.
  Springer (1993). \doi{10.1007/3-540-57899-4\_45}

\bibitem{ChandraM77}
Chandra, A.K., Merlin, P.M.: Optimal implementation of conjunctive queries in
  relational data bases. In: Hopcroft, J.E., Friedman, E.P., Harrison, M.A.
  (eds.) Proceedings of the 9th Annual {ACM} Symposium on Theory of Computing,
  May 4-6, 1977, Boulder, Colorado, {USA}. pp. 77--90. {ACM} (1977).
  \doi{10.1145/800105.803397}

\bibitem{DalmauKV02}
Dalmau, V., Kolaitis, P.G., Vardi, M.Y.: Constraint satisfaction, bounded
  treewidth, and finite-variable logics. In: Hentenryck, P.V. (ed.) Principles
  and Practice of Constraint Programming - {CP} 2002, 8th International
  Conference, {CP} 2002, Ithaca, NY, USA, September 9-13, 2002, Proceedings.
  Lecture Notes in Computer Science, vol.~2470, pp. 310--326. Springer (2002).
  \doi{10.1007/3-540-46135-3\_21}

\bibitem{DattaKMSZ18}
Datta, S., Kulkarni, R., Mukherjee, A., Schwentick, T., Zeume, T.: Reachability
  is in {D}yn{F}{O}. J. {ACM}  \textbf{65}(5),  33:1--33:24 (2018).
  \doi{10.1145/3212685}

\bibitem{DattaMSVZ19}
Datta, S., Mukherjee, A., Schwentick, T., Vortmeier, N., Zeume, T.: A strategy
  for dynamic programs: Start over and muddle through. Log. Methods Comput.
  Sci.  \textbf{15}(2) (2019). \doi{10.23638/LMCS-15(2:12)2019}

\bibitem{DongST95}
Dong, G., Su, J., Topor, R.W.: Nonrecursive incremental evaluation of datalog
  queries. Ann. Math. Artif. Intell.  \textbf{14}(2-4),  187--223 (1995).
  \doi{10.1007/BF01530820}

\bibitem{fagin1983degrees}
Fagin, R.: Degrees of acyclicity for hypergraphs and relational database
  schemes. Journal of the ACM (JACM)  \textbf{30}(3),  514--550 (1983).
  \doi{10.1145/2402.322390}

\bibitem{FederV98}
Feder, T., Vardi, M.Y.: The computational structure of monotone monadic {SNP}
  and constraint satisfaction: {A} study through datalog and group theory.
  {SIAM} J. Comput.  \textbf{28}(1),  57--104 (1998).
  \doi{10.1137/S0097539794266766}

\bibitem{GeladeMS12}
Gelade, W., Marquardt, M., Schwentick, T.: The dynamic complexity of formal
  languages. {ACM} Trans. Comput. Log.  \textbf{13}(3),  19:1--19:36 (2012).
  \doi{10.1145/2287718.2287719}

\bibitem{gottlob2001complexity}
Gottlob, G., Leone, N., Scarcello, F.: The complexity of acyclic conjunctive
  queries. J. {ACM}  \textbf{48}(3),  431--498 (2001).
  \doi{10.1145/382780.382783}

\bibitem{GreenOW15live}
Green, T.J., Olteanu, D., Washburn, G.: Live programming in the {L}ogic{B}lox
  system: {A} {M}eta{L}ogi{QL} approach. Proc. {VLDB} Endow.  \textbf{8}(12),
  1782--1791 (2015). \doi{10.14778/2824032.2824075}

\bibitem{Grohe07}
Grohe, M.: The complexity of homomorphism and constraint satisfaction problems
  seen from the other side. J. {ACM}  \textbf{54}(1),  1:1--1:24 (2007).
  \doi{10.1145/1206035.1206036}

\bibitem{10.5555/310709.310737}
Gupta, A., Mumick, I.S.: Maintenance of Materialized Views: Problems,
  Techniques, and Applications, p. 145–157. MIT Press, Cambridge, MA, USA
  (1999)

\bibitem{HellN90}
Hell, P., Nesetril, J.: On the complexity of \emph{H}-coloring. J. Comb.
  Theory, Ser. {B}  \textbf{48}(1),  92--110 (1990).
  \doi{10.1016/0095-8956(90)90132-J}

\bibitem{Koch10incr}
Koch, C.: Incremental query evaluation in a ring of databases. In: Paredaens,
  J., Gucht, D.V. (eds.) Proceedings of the Twenty-Ninth {ACM}
  {SIGMOD-SIGACT-SIGART} Symposium on Principles of Database Systems, {PODS}
  2010, June 6-11, 2010, Indianapolis, Indiana, {USA}. pp. 87--98. {ACM}
  (2010). \doi{10.1145/1807085.1807100}

\bibitem{Libkin04}
Libkin, L.: Elements of Finite Model Theory. Springer (2004).
  \doi{10.1007/978-3-662-07003-1}

\bibitem{muoz_et_al}
Mu{\~n}oz, P., Vortmeier, N., Zeume, T.: {Dynamic Graph Queries}. In: Martens,
  W., Zeume, T. (eds.) 19th International Conference on Database Theory (ICDT
  2016). Leibniz International Proceedings in Informatics (LIPIcs), vol.~48,
  pp. 14:1--14:18. Schloss Dagstuhl--Leibniz-Zentrum fuer Informatik, Dagstuhl,
  Germany (2016). \doi{10.4230/LIPIcs.ICDT.2016.14}

\bibitem{patnaikI97}
Patnaik, S., Immerman, N.: Dyn-{FO}: {A} parallel, dynamic complexity class. J.
  Comput. Syst. Sci.  \textbf{55}(2),  199--209 (1997).
  \doi{10.1006/jcss.1997.1520}

\bibitem{SchwentickVZ20}
Schwentick, T., Vortmeier, N., Zeume, T.: Sketches of dynamic complexity.
  {SIGMOD} Rec.  \textbf{49}(2),  18--29 (2020). \doi{10.1145/3442322.3442325}

\bibitem{Venkateswaran91}
Venkateswaran, H.: Properties that characterize {LOGCFL}. J. Comput. Syst. Sci.
   \textbf{43}(2),  380--404 (1991). \doi{10.1016/0022-0000(91)90020-6}

\bibitem{conferenceversion}
Vortmeier, N., Kokkinis, I.: The dynamic complexity of acyclic hypergraph
  homomorphisms, accepted for publication at the International Workshop on
  Graph-Theoretic Concepts in Computer Science, WG 2021.

\bibitem{VortmeierZ20}
Vortmeier, N., Zeume, T.: Dynamic complexity of parity exists queries. In:
  Fern{\'{a}}ndez, M., Muscholl, A. (eds.) 28th {EACSL} Annual Conference on
  Computer Science Logic, {CSL} 2020, January 13-16, 2020, Barcelona, Spain.
  LIPIcs, vol.~152, pp. 37:1--37:16. Schloss Dagstuhl - Leibniz-Zentrum
  f{\"{u}}r Informatik (2020). \doi{10.4230/LIPIcs.CSL.2020.37}

\bibitem{yannakakis1981algorithms}
Yannakakis, M.: Algorithms for acyclic database schemes. In: Very Large Data
  Bases, 7th International Conference, September 9-11, 1981, Cannes, France,
  Proceedings. pp. 82--94. {IEEE} Computer Society (1981),
  \url{https://dl.acm.org/doi/10.5555/1286831.1286840}

\bibitem{Zeume17}
Zeume, T.: The dynamic descriptive complexity of k-clique. Inf. Comput.
  \textbf{256},  9--22 (2017). \doi{10.1016/j.ic.2017.04.005}

\end{thebibliography}

\end{document}